\documentclass[10pt, a4paper]{article}
\usepackage{amsthm, amssymb, amsmath}
\usepackage{verbatim, float}
\usepackage{mathrsfs}
\usepackage{graphics}
\usepackage[usenames,dvipsnames]{pstricks}
\usepackage{epsfig}
\usepackage{pst-grad} 
\usepackage{pst-plot} 
\usepackage{cases}
\usepackage{algorithm,algorithmic}

\usepackage[top=2.4cm, bottom=2.4cm, left=2.5cm, right=2.5cm]{geometry}
\usepackage{url}

\theoremstyle{plain}
\newtheorem{theorem}{Theorem}
\newtheorem{lemma}{Lemma}

\theoremstyle{definition}

\newcommand{\sizeof}[1]{\left\lvert{#1}\right\rvert}
\newcommand{\Na}{N_{\alpha}}
\newcommand{\caze}[2]{\textbf{Case {#1}:} \textit{#2}}
\newcommand{\st}{\colon\,}
\newcommand{\A}{{\mathcal A}}
\newcommand{\B}{{\mathcal B}}
\newcommand{\C}{{\mathcal C}}
\newcommand{\D}{{\mathcal D}}

\renewcommand{\O}{{\mathcal O}}

\newcommand{\nin}{\noindent}

\newcommand{\tvG}{{\theta_{\circ}(G)}}
\newcommand{\te}{{\theta_{|}}}
\newcommand{\teG}{{\theta_{|}(G)}}
\newcommand{\tet}{{\theta_{|,\triangle}}}
\newcommand{\tetG}{{\theta_{|,\triangle}(G)}}

\newcommand{\tuvw}{{\theta_{|,\triangle}^{\{u,v,w\}}}}

\newcommand{\turan}{Tur\'{a}n}

\begin{document}
\title{Motif Clustering and Overlapping Clustering for Social Network Analysis\footnotemark[1]}
\author{ Pan~Li, Hoang~Dau, Gregory~Puleo and Olgica~Milenkovic\footnotemark[2] }
\footnotetext[1]{A shorter version of this will appear in Proc. 2017 IEEE Conference on Computer Communications (INFOCOM). The research is supported in part by the NSF grant CCF-1029030 and the NSF STC for Science of Information.} 
   \footnotetext[2]{The authors are with the Coordinated Science Laboratory, Department of Electrical and Computer Engineering, University of Illinois at Urbana-Champaign (email: panli2@illinois.edu, hoangdau@illinois.edu, gjp0007@illinois.edu, milenkov@illinois.edu).
   }
\maketitle
\begin{abstract}
Motivated by applications in social network community analysis, we introduce a new clustering paradigm termed~\emph{motif clustering}. Unlike classical clustering, motif clustering aims to minimize the number of clustering errors associated with both edges and certain higher order graph structures (motifs) that represent ``at
omic units'' of social organizations. 
Our contributions are two-fold: We first introduce \emph{motif correlation clustering}, in which the goal is to \emph{agnostically} partition the vertices of a weighted complete graph so that certain predetermined ``important'' social subgraphs mostly lie within the same cluster, while ``less relevant'' social subgraphs are allowed to lie across clusters. We then proceed to introduce the notion of \emph{motif covers}, in which the goal is to cover the vertices of motifs via the smallest number of (near) cliques in the graph. Motif cover algorithms provide a natural solution for overlapping clustering and they also play an important role in latent feature inference of networks. For both motif correlation clustering and its extension introduced via the covering problem, we provide hardness results, algorithmic solutions and community detection results for two well-studied social networks.
\end{abstract}
\section{Introduction} \label{sec:intro}

The problem of clustering vertices of graphs has received significant attention in physics, biology and 
computer science due to the fact that it reveals important properties regarding the community structure of the underlying networks~\cite{jain1988algorithms,benson2016higher}. Clustering may result in a partition of the vertices, or a decomposition of the vertex set into intersecting subsets that are often referred to as overlapping communities~\cite{barthelemy2001np}. 
In most machine learning settings, one focuses on spectral clustering methods~\cite{von2007tutorial} and assumes that the number of clusters 
or an upper bound on the number of clusters is known beforehand, or that the parameters of the model may be learned efficiently~\cite{pelleg2000x,K-means}. 
On the other hand, some clustering methods proposed in the computer science literature~\cite{BBC1} adopt agnostic approaches that often result in computationally hard problems that may only be solved approximately~\cite{DEFI1}. The algorithms used to perform clustering 
range from greedy and iterative methods to semidefinite and linear programs accompanied by rounding techniques~\cite{ACN2,CGW1}, and may be implemented in parallel~\cite{pan2015parallel}. 

One important, yet highly overlooked aspect of community detection is that in order to capture relevant social phenomena, one has to understand higher order interactions of entities in the community. These higher order interactions correspond to induced subgraphs of the social networks, and as such, should be considered as ``atomic units'' of the graph. Clearly, edges represent one such unit, as they capture pairwise interactions, but almost equally important entities are triangles, which are known to be social and biological network motifs (i.e., subgraphs that appear with frequency exceeding the one predicted through certain random models). Hence, when clustering vertices in a graph it may be important to place a motif such as a triangle within the same cluster, rather than between clusters. Related problems have been studied in different contexts and with different motivations under the name of \emph{hypergraph clustering} in a fairly limited number of contributions~\cite{benson2016higher, zhou2006learning,leordeanu2012efficient,agarwal2005beyond,kim2011higher,angelini2015spectral}. Almost all of the methods proposed for this particular setting are heuristics that are constrained by knowledge of the problem parameters. Furthermore, the methods appear hard to interpret in one unified framework that involves both nonoverlapping and overlapping clusters, and tend to use spectral techniques which often do not come with general analytical guarantees. None of the methods treats hyperedges of different sizes as having different relevance, as the hyperedges are usually not seen as entities that arise from subgraphs of a social graph. In addition, none of the hypergraph clustering methods extends to overlapping clustering.

Here, we take a very general and broad new approach to hypergraph clustering by building on the ideas behind classical \emph{correlation clustering}~\cite{BBC1}, which may be succinctly described as follows: One is given a graph and, for some pairs of vertices, one is also given a \emph{quantitative assessment} of whether the objects are \emph{similar} or \emph{dissimilar}. The goal is to partition the vertices of the graph so that similar vertices tend to aggregate within clusters and dissimilar vertices tend to belong to different clusters. Instead of looking at the problem of clustering individual vertices, we focus our attention on simultaneously clustering subgroups of vertices forming specific, prescribed subgraphs  in the graph. We impose weights on the cost of subgraph clustering, which allow one to assess the penalty of placing the subgraph across clusters or within one cluster, thereby taking structural relevance into account. Based on ideas behind an overlapping correlation clustering technique suggested in~\cite{BonchiGionisUkkonenICDM2011}, we also develop motif correlation clustering techniques for overlapping community detection. In this setting, the goal is to cover all motifs by the smallest number of cliques or near cliques in the graph. Our interpretation also gives rise to a new direction in the field of intersection graph theory~\cite{erdos1966representation} and may be used for latent feature inference~\cite{TsourakakisWWW2015}. For succinctness, we mostly focus our attention on two types of motifs only, edges and triangles. The results described for edges and triangles may be extended to account for higher order structures. 

The paper is organized as follows. In Section~\ref{sec:overview}, we describe correlation clustering and overlapping correlation clustering. Section~\ref{sec:mcc} introduces our new motif correlation clustering paradigm. There, we show that the problem of interest is NP-complete and describe a constant approximation algorithm for clustering based on a linear programming (LP) relaxation followed by rounding. We then proceed to introduce the overlapping motif correlation clustering problem in Section~\ref{sec:omcc}, prove that it is NP-complete and provide some theoretical results on the largest number of clusters needed for the coverings. We also introduce a heuristic simulated annealing algorithm for overlapping clustering that performs well in practice and generalizes the work in~\cite{TsourakakisWWW2015}. We conclude with Section~\ref{sec:simulations}, which contains simulation results for two networks with ground-truth community structures, illustrating the concepts of motif and overlapping motif clustering. 
Large scale network analysis is relegated to a companion paper.  

\section{Correlation Clustering and Overlapping Correlation Clustering} \label{sec:overview}

There are two dual formulations of the correlation clustering optimization problem: \emph{MinDisagree} and \emph{MaxAgree}. 
In both cases, one is given a graph whose vertices are to be clustered, with each edge labeled so as to indicate 
whether the endpoint vertices are to lie within the same cluster or not. For the MinDisagree version of the problem, one aims to minimize the number of erroneously placed edges (pairs of vertices), while for the dual MaxAgree version, one seeks to maximize the total number of correctly placed edges. Finding an optimal solution to either problem is NP-complete, but the MinDisagree version of the problem is harder to approximate. As from the perspective of experimental design and quality of service erroneously clustered vertices are often more costly than correctly clustered ones, a large body of work has focused on the MinDisagree version of the problem~\cite{BBC1}. Unfortunately, the MinDisagree problem remains hard even when the input graph is
complete~\cite{bansal2004correlation}. For complete graphs, several constant approximation randomized~\cite{ACN2} and deterministic~\cite{near-optimal} algorithms are known. When the graph is allowed to be arbitrary, the best known approximation ratio is $O(\log n)$~\cite{DEFI1}.

Some variants of correlation clustering allow for including fractional edge weights into the problem formulation, with each edge endowed with a ``similarity'' and ``dissimilarity'' weight: If the edge is placed across clusters, the edge is charged its similarity weight, and if the edge is placed within the same cluster, the edge is charged its dissimilarity weight. The MinDisagree clustering goal is to minimize the overall vertex partitioning weight (cost). Clearly, if the weights are unrestricted, not all instances of the weighted clustering problem may be efficiently approximated. Hence, most of the work has focused on so-called 
\emph{probability weights}~\cite{BBC1}. The classical probability weights correlation clustering problem formulation for a weighted graph $G=(V(G),E(G))$ may be written as:
  \[\begin{aligned}
    & \underset{x}{\text{minimize}}
    & & \left[\sum_{e \in E(G)}(w_e x_e + (1-w_e) (1-x_e))\right] \\
    & \text{subject to}
    & & x_{uv} \leq x_{uz} + x_{zv} \quad{\text{(for all distinct $u,v,z \in V(G)$)}} & \\
    &&&  x_{e} \in \{{0,1\}} \quad{\text{(for all $e \in E(G)$)}} &
      \end{aligned} \label{fig:LP}
    \]
Here, the variables $x_e$ are indexed by edges $e$ and interpreted as follows: $x_e = 1$ means that the endpoints of $e$ lie in different clusters while $x_e = 0$ means that the endpoints of $e$ lie in the same cluster. The cost of placing $e$ across clusters is $0 \leq w_e \leq 1$, while the cost of placing $e$ within the same cluster equals $1-w_e$. The triangle inequality $x_{uv} \leq x_{uz} + x_{zv}$ captures the fact that if two edges with vertices $uz$ and $zv$ are in the same cluster, then the edge with vertices $uv$ should also belong to the same cluster.

As the problem described in the former setting is hard~\cite{BBC1}, a standard approach is to relax the constraint $x_{e} \in \{{0,1\}}$ to $x_{e} \in [0,1]$, and then round the fractional values $x_e$~\cite{CGW1}.

An equivalent formulation of the correlation clustering problem, which naturally extends to an overlapping community setting, may be stated as follows~\cite{TsourakakisWWW2015, BonchiGionisUkkonenICDM2011}.

As before, one is given a graph $G=(V(G),E(G))$, $|V|=n$, and a similarity weight function $w: V \times V \to [0,1],$ as well as a sufficiently large set of labels (features) $L$. The labels will give rise to the vertex partition by grouping all vertices with the same label into one cluster. Correlation clustering reduces to finding a labeling $\ell: V \to L$ which minimizes 
        \begin{equation}
        \sum_{uv \in E(G), \, \ell(u)=\ell(v)} (1-w_{uv})+\sum_{uv \in E(G), \, \ell(u) \neq \ell(v)} w_{uv}. \notag
        \end{equation}
A simple extension of this formulation for the case of overlapping clusters is to assign a set of labels to each vertex, rather than one label only. This implies multiple cluster membership for some vertices. In this setting, let $A,B$ denote sets and let $H(A,B)$ be some chosen set similarity function. Furthermore, let $\ell$ be a set labeling function. The goal of overlapping clustering now becomes to find a labeling function $\ell: V \to \mathcal{P}(L)$, where $\mathcal{P}(L)$ denotes the power set of $L$, that minimizes
         \begin{equation}
        \sum_{uv \in E(G)} |H(\ell(u),\ell(v))-w_{uv}|. \notag
        \end{equation}
The objective function takes different forms depending on the chosen set similarity function $H(A,B)$. If $H(A,B)=1$ for $A \cap B \neq \emptyset$, and zero otherwise, overlapping correlation clustering reduces to an instance of the \emph{intersection representation problem} from graph theory~\cite{erdos1966representation}. An intersection representation of a finite, undirected graph $G=(V(G),E(G))$ is an assignment of subsets $\mathcal{I}_u$ of a finite, sufficiently large ground set $\mathcal{F}$, to vertices $u \in V$ such that $(u,v) \in E$ if and only if $\mathcal{I}_u \cap \mathcal{I}_v \neq \emptyset$. The smallest cardinality of the ground set $\mathcal{F}$ needed to properly represent the graph is known as the \emph{intersection number of the graph}. It is known that the the intersection number of a graph equals its edge clique cover number, i.e., the smallest number of cliques in the graph needed to cover all edges in the graph~\cite{ErdosGoodmanPosa1966}. It is clear that given an intersection representation 
of the graph, the set of vertices that are assigned to a particular clique may be seen as sharing one feature. This is why the intersection representation of a graph is often used for latent feature inference.

An example of an intersection representation of a graph over the smallest ground set $\mathcal{F}=\{{1,2,3\}}$ is shown in Figure~\ref{fig:C4}.    
\begin{figure}[htb]
\centering
\includegraphics[scale=1]{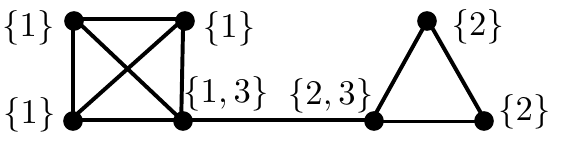}
\caption{An intersection representation of a graph using three features $\{1,2,3\}$.} 
\label{fig:C4}
\end{figure}  
\section{Motif Correlation Clustering} \label{sec:mcc}
\vspace{-0.08in}
We depart from the classical correlation clustering problem by considering a new setting in which one is allowed to assign probability weights to both edges and arbitrary small induced subgraphs in the graph and then perform the clustering so as to minimize the overall cost of both edge and higher motif placements. The described method focuses on weighted undirected and complete graphs, but despite these apparent topological limitations, it allows one to handle motifs in \emph{both directed or incomplete graphs} by encoding information about the ``relevance'' of directed or incomplete subgraphs of the graph via the assigned similarity/dissimilarity weights. For example, if in a directed graph the only motifs of interest are \emph{feedforward triangles}, only those $3$-tuples of vertices corresponding to these directed triangle structures will be assigned large similarity weights in the undirected complete graph and hence encouraged to lie within clusters. If triangles are deemed to be relevant, $3$-tuples corresponding to triangles in the original graph are assigned large similarity weight. 
Consequently, motif correlation clustering may be used in applications as diverse as layered flow analysis in a information networks, anomaly detection in communication networks or for determining hierarchical community structure detection in gene and neuronal regulatory networks~\cite{benson2016higher}. 
\subsection{Problem Formulation} \label{sec:mcc-problem}

As already remarked in the motivation section, any incomplete graph may be converted into a weighted complete graph by assigning weights to the $k$-tuples of vertices of the complete graph so as to capture the presence of both edges and non-edges and higher structural units in the initial graph. For example, a nonedge in the initial graph may be assigned a similarity weight $\leq 1/2$, thereby not (significantly) biasing the clustering objective function towards any particular solution. Similarly, edges in the initial graph may be assigned similarity weight $1$, thereby strongly forcing their corresponding vertices to cluster within the same community. 
The same logic may be applied to directed graphs as well. This is why throughout the rest of paper we assume that the graphs of interest $G(V,E)$ are \emph{undirected, weighted complete graphs} with vertex set $V$ of cardinality $n$ and edge set $E$ of cardinality $\binom{n}{2}$. 
We also use the symbol $K$ to denote an arbitrary element of $\Pi(V)$, the set of all $k$-tuples of $V$ such that $2 \leq k \leq n$.
Suppose next that to each $K \in \Pi(V)$ we assign a pair of non-negative values, $(w_K,1-w_K)$, respectively.
 The weights $w_{K}$ and $1-w_{K}$ indicate the respective costs of placing the vertices in $K$ across and within the cluster, respectively. Therefore, to enable motif clustering, the similarity weights $w_K$ of the tuples that constitute motifs in the initial graph should be large. The goal is to solve the following MinDisagree version of the motif clustering problem, termed Mixed Motif Correlation Clustering (MMCC): Fix multiple motif graphs in the initial graph of possibly different sizes that belong to the set $S=\{{k_1< k_2< \ldots < k_p\}},$ 
and seek a vertex partition $C=(C_1,\ldots,C_s)$, $s \geq 1$, that solves
\begin{align}\label{eq:mmcc}
\text{(MMCC)} \;\; \min_{C} \sum_{t=1}^p \;  &\lambda_{t} \, \sum_{K\subseteq C_i \, \text{ for some $i$}, \, |K|=k_t}(1-w_{K})
\notag
\\
+ &\lambda_{t} \, \sum_{K\not \subseteq C_i \, \text{for all $i$}, \, |K|=k_t}w_{K}.
\end{align}
Here, $\lambda_{t} \geq 0$ denotes the relevance factor of motifs of size $k_t$. Note that by choosing $\lambda_1=1$ for edges and setting all other relevance factors to zero, we arrive at the classical correlation clustering formulation. 

To explain the underlying clustering approach, we henceforth assume that $p=2$, and that the motifs are of size 
two and three (i.e., edges and triangles). For simplicity of exposition, in our theoretical analysis we fix the relevance factors to $\lambda_1=\lambda_2=1$ (and set all other relevance factors to zero). In the subsequent simulations, we allow the $3$-tuple relevance factor $\lambda_2$ to change in order to explain practical community detection findings. 

It may be shown 
that the edge/triangle MMCC problem is NP-complete by using a reduction from the Partition into Triangles problem~\cite{creignou1995class} (The proof of this result may be found in Appendix~\ref{app:MMCCHard}. We only outline the triangle clustering proof, as the edge/triangle case is a simple consequence of this result and the one pertaining to classical edge correlation clustering). Hence, we focus on developing (constant) approximation algorithms for the underlying problem. 

As before, let $S=\{{2,3\}}$ be the set of motif sizes, and let $E(V)$ 
and $\mathcal{T}(V)$ stand for the set of all edges and $3-$tuples of $V$, respectively. Let $T$ stand for a generic $3$-tuple and 
let $x_T$ denote the indicator of the event that the vertices in the tuple are split among clusters. Furthermore, let $x_e,$ $e\in E(G),$ denote the indicator of the event that the pair of vertices corresponding to $e$ belongs to different clusters 
(i.e., $x_e=0$ if $e=(vw)$ and $v$ and $w$ belong to the same cluster, and $x_e=0$ otherwise). As for the general MMCC problem, 
we let $w_e$ denote the similarity weight of a $2$-tuple, and $w_T$ denote the similarity weight of a $3$-tuple $T$. Recall that $2$-tuples and $3$-tuples that correspond to edges and triangles in the initial graph will be weighted differently than $2-$tuples and $3$-tuples corresponding to nonedges and nontriangles.

By relaxing the indicator variable constraints to $x_{T},x_e \in [0,1]$, we arrive at the following LP problem formulation for the MMCC problem: 
\begin{align} 
\min_{\{x_e,x_T\}}\, &\sum_{e \in E(V)} w_e \, x_e+(1-w_e)(1-x_e)+ \\
&\sum_{T\in\mathcal{T}(V)} w_T \, x_T+(1-w_T)(1-x_T) \;\; \text{s.t. } \nonumber \\
& a) \; x_T\geq x_e  \; \text{(for all $e\in E(V)$ and $e \subset T \in \mathcal{T}(V)$)}, \nonumber \\
& b)\; x_T\leq \frac{1}{2}\sum_{u,v\in T} x_{uv}, \;\; x_T\leq 1 \; \text{(for all $T \in \mathcal{T}(V)$)},
\nonumber \\
&x_e\geq 0 \;\; \text{(for all $e \in E(V)$)}, \nonumber\\
&c) \;x_{vw}\leq x_{uv}+x_{uw} \; \text{(for all distinct $u,v,w \in V)$}. \nonumber
\end{align}
Here, the constraints are to be interpreted as follows: The constraint a) ensures that if an edge lies across clusters, all triangles $T$ including that edge have to lie across clusters. The constraint b) ensures that if all three edges of a triangle lie within a cluster, then the corresponding triangle has to lie within the same cluster, and if a triangle is split, at least two edges lie across clusters. The constraint c) implies that placing two adjacent edges of a constituent triangle within a cluster leads to placing a third adjacent edge into the same cluster.

The rounding method accompanying this LP is described in Algorithm~1, with the parameters $\alpha,\beta$ set to $1/\max\{{S\}}=1/3$. Except for a different scaling scheme, the proposed rounding procedure essentially follows the classical region growing method of~\cite{CGW1}, but imposes nontrivial analytical challenges when coupled with our new LP formulation.
\begin{table}[htb]
\centering
\begin{tabular}{l}
\hline
\label{alg:rounding2}
\textbf{Algorithm 1 Rounding Procedure with parameters $\alpha,\beta = \frac{1}{3}$}\\
\textbf{Initialization:} $\mathcal{V} = V(G)$; \\
\ 1: \textbf{repeat} \\
\ 2: \quad Choose an arbitrary ``pivot vertex'' $u$ in $\mathcal{V}$;\\
\ 3: \quad Let $\Na(u)=\{v\in \mathcal{V} -\{u\}: x_{uv}\leq \alpha\}$;\\
\ 4: \quad \textbf{if} $\sum_{v\in \Na(u)} x_{uv}>  \beta\alpha|\Na(u)|$\\
\ 5: \quad \quad Output the singleton cluster $\{u\}$;\\
\ 6: \quad \textbf{else} \\
\ 7: \quad \quad Output the cluster $C=\Na(u)\cup\{u\}$; \\
\ 8: \quad  Let $\mathcal{V}=\mathcal{V}-C$ \\
\ 9: \textbf{until} $|\mathcal{V}|< 3$;\\
\textbf{Output:} Output all sets $C$;\\
\hline
\end{tabular}
\end{table}

\begin{theorem} For the parameter choices $\alpha,\beta = 1/3$, the LP and rounding algorithm provides an $1/(\alpha\beta)=9$-approximation for the MMCC problem. \end{theorem}
\begin{proof} It may be shown that proving approximation guarantees for clustering of multiple motifs may be reduced to proving corresponding results for the \emph{largest size motif only}, which in this case corresponds to a $3$-tuple. The performance guarantees for triangle clustering are established in Appendix~\ref{sec:app}.  \end{proof}

\vspace{-0.1in}
The number of constraints in the LP solver for the general MMCC problem equals $\mathcal{O}(n^k)$, where $k$ is the size of the largest motif considered. For edge and triangle motifs, this results in a number of constraints roughly equal to $\mathcal{O}(n^3)$. To speed up computations and make the algorithm scalable for large networks one may utilize the sparsity of the constraints and efficient approximate LP solvers, such as those based on parallel stochastic-coordinate-descent~\cite{sridhar2013approximate}. The aforementioned LP solver offers order of magnitude improvements in execution speed compared to the Cplex LP solver.

Consider next the following alternative formulation of the motif correlation clustering problem. To simplify our explanation, we consider motifs involving $3$-tuples only, which we generically denote by $T=\{{a,b,c\}}, a,b,c \in V$ (The problem formulation below may be easily generalized to include any combination of motifs, analog to what was described for correlation clustering in Equation~(1)). Using the notion of vertex labels described in the introduction, the objective function of the $3$-tuple correlation clustering problem may be rewritten as: 
\begin{equation}
\sum_{T: \, \ell(a)=\ell(b)=\ell(c)} (1-w_T)+\sum_{T: \,(\ell(a)=\ell(b)=\ell(c))'} w_T. \notag
\end{equation}
Here, with a slight abuse of notation, $\mathcal{E}'$ stands for the complement of the event $\mathcal{E}$, which in this case indicates that at least two vertices in the $3$-tuple have different labels. This formulation also has a natural interpretation in the context of hypergraph clustering and it is straightforward to formulate a similar objective involving edges and triangles, which equals a correlation clustering formulation for hypergraphs with two types of edges. Similarly to what was described for correlation clustering, one may extend the triangle clustering paradigm into an overlapping clustering paradigm by introducing a set similarity function $H$, which this time operates on three sets, 
say $A,B,C$ so that
\begin{equation}
        \sum_{T: \, a \neq b \neq c} |H(\ell(a),\ell(b),\ell(c))-w_T|, \notag
        \end{equation}
where as before $T=\{{a,b,c\}}$. Note that if we choose a set similarity function of the form $H(A,B,C)=1$ if $A \cap B \cap C \neq \emptyset$, and $H(A,B,C)=0$ otherwise, we arrive at the (new) problem of \emph{triangle clique cover}. 
This type of cover may be easily formulated to include any higher order graph structure, and is the focal point of the analysis presented in the next section.
             \vspace{-0.1in}
\section{Overlapping Motif Correlation Clustering via Edge-Triangle Clique Covers of Graphs}
\label{sec:omcc}

Recall that an edge clique cover (ECC) of an undirected graph $G$ is a set of cliques of $G$ 
that collectively covers all of its \emph{edges}, and that the edge clique cover number (intersection number) of the graph $\teG$ equals the minimum number of cliques in any ECC.  We introduce the concept of a motif cover of a graph $G$, which is a set of cliques of $G$ 
that collectively covers all the chosen \emph{motif structures} in $G$. In particular, we focus on the new paradigm of \emph{edge-triangle clique cover} (ETCC) of a graph,
which is a set of cliques in the graph that collectively covers all edges and triangles in the graph. The smallest such number of cliques $\tetG$ will be referred to as the \emph{edge-triangle clique cover number}. Clearly, the edge-triangle clique cover formulation represents nothing more than a combinatorial interpretation of the motif correlation clustering problem outlined in the previous section, with each shared element of the sets $A,B,C$ describing a clique/cluster/community. An example illustrating the concepts of ECC and ETCC is shown in Figure~\ref{fig:toy_example}. 

It was shown in~\cite{Orlin1977} that determining $\teG$ is an NP-complete problem. The idea is to reduce the problem of determining the vertex clique cover number $\tvG$, which is known to be NP-complete, to the problem of determining $\teG$.  This result may be generalized to show that determining $\tetG$ is an NP-complete problem by using ideas from~\cite{KouStockmeyerWong78} and by reducing the problem 
of determining $\teG$ to the problem of determining $\tetG$. Details of the proof may be found in Appendix~\ref{app:ETCHard}.

As the edge-triangle clique cover problem is NP-complete, we focus on developing a simple simulated annealing algorithm for finding an approximate edge-triangle cover. One of the problem parameters of the annealing algorithm is the number (or an upper bound on the number) of near-cliques or cliques needed to cover the edges and triangles in the graph\footnote{Note that, in general, $M$ does not have to be a bound on the edge-triangle clique cover number, as one may want to have communities that are not necessarily cliques. Choosing the parameter $M$ to be smaller than the edge-triangle clique cover number will force smaller clusters to be lumped together.}. We derive one such upper bound by  a nontrivial generalizations of upper bounds on the intersection number derived in~\cite{ErdosGoodmanPosa1966,Alon1986} for the case of the edge-triangle clique cover number. 
\vspace{-0.05in}

\subsection{A Simulated Annealing Algorithm}
\label{subsec:simulated_annealing}


As the ETCC is hard to solve exactly, we seek an approximate empirical algorithm that may perform the covering efficiently on large scale networks. Such an approach was also proposed in the context of computing approximations for intersection numbers in~\cite{BonchiGionisUkkonenICDM2011,TsourakakisWWW2015}. There, given a fixed number of features (clusters, communities) $M$, the algorithm assigns subsets of features to the vertices of the graph in a way that maximizes a certain \emph{score}, which for simplicity may be taken to equal the number of pairs 
$(u,v) \in V \times V$ that satisfy the previously described set intersection conditions. Once a feature assignment with a large score is found, each set of vertices assigned one particular feature is treated as a
cluster, or equivalently, a community. As each vertex can be assigned more than one features, the output communities are naturally overlapping.
Furthermore, as the solution is only approximate, the communities do not necessarily correspond to cliques but to dense subgraphs, which is actually a desirable property for real world network community detection, where cliques as communities may be rather unrealistic. For example, in Facebook friendship networks, a group of people sharing one common feature -- say, having graduated from the same school -- does not necessarily imply pairwise Facebook friendship.
\begin{figure}[htb]
\centering
\includegraphics[scale=0.9]{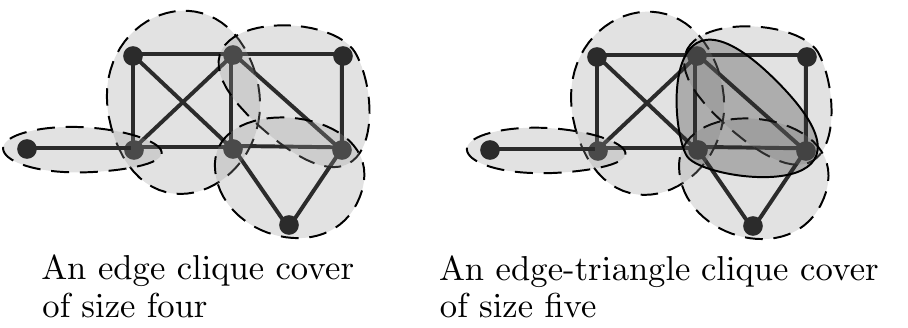}
\caption{An edge clique cover and an edge-triangle clique cover. Note that the edge clique cover does not form an edge-triangle clique cover because it does not cover the
middle (darker shaded) triangle.}
\label{fig:toy_example}
\end{figure}

In what follows, we describe a new simulated annealing algorithm for detecting overlapping communities
that takes into consideration both edges and triangles. We recall that an \emph{edge-triangle intersection representation} of a graph requires that two vertices be adjacent if and only if they share a common feature, and similarly, three vertices $u$, $v$, and $w$ form a triangle
if and only if they all share at least one common feature. We henceforth refer to these conditions as the \emph{Edge-Triangle Intersection Condition},
which essentially guarantees that two or three vertices belong to a common \emph{community} if and only if they are pairwise adjacent.  
Given an estimated number of communities $M$, the objective function may be written as follows:
\begin{equation}
\label{eq:SA}
\begin{split}
\max_{A \colon V \to 2^{[M]}} 
&\ \alpha_{\text{e}}\sum_{(u,v) \in E} \chi(A_u \cap A_v \neq \varnothing)\\ 
&+ \alpha_{\text{ne}}\sum_{(u,v) \notin E} \big(1 - \chi(A_u \cap A_v \neq \varnothing)\big)\\ 
&+ \alpha_{\text{t}}\sum_{(u,v,w) \in \mathcal{T}} \chi(A_u \cap A_v \cap A_w \neq \varnothing)\\ 
&+ \alpha_{\text{nt}}\sum_{(u,v,w) \notin \mathcal{T}} \big(1 - \chi(A_u \cap A_v \cap A_w \neq \varnothing)\big),
\end{split}
\end{equation}
where $A_u \subseteq \{1,2,\ldots,M\}$ denotes the set of features of the ground set assigned to the vertex $u$, $\mathcal{T}$ denotes the set of triangles of $G$,
and $\chi(\mathcal{C}) = 1$ if the clause $\mathcal{C}$ is correct and $\chi(\mathcal{C}) = 0$ otherwise. The parameters $\alpha$ essentially represent the rewards of edges, nonedges, triangles and nontriangles satisfying the edge-triangle intersection rules described above.

When $\alpha_{\text{e}} = \alpha_{\text{ne}} = \alpha_{\text{t}} = \alpha_{\text{nt}} = 1,$ the solution to the optimization problem~\eqref{eq:SA} corresponds to an approximate edge-triangle intersection representation of $G$ with highest \emph{score}, which is defined as the number of pairs $(u,v)$ and $3$-tuples $(u,v,w)$ that have feature sets satisfying the Edge-Triangle Intersection Condition. In sparse networks, the number of edges can be much smaller than the number of non-edges and the number of non-triangles. Therefore, it is desirable to tune the rewards as follows:
\begin{equation}
\label{eq:weights}
\alpha_{\text{e}} = 1, \alpha_{\text{ne}} = \frac{|E|}{\binom{n}{2}-|E|},
\alpha_{\text{t}} = \frac{|E|}{|\mathcal{T}|}, \alpha_{\text{nt}} = \frac{|E|}{\binom{n}{3}-|\mathcal{T}|},
\end{equation}
where the sums are normalized according to their numbers of terms. 

Let $s(A)$ denote the \emph{normalized} score of the feature assignment $A \colon V \to
2^{[M]}$ with respect to the weights given in \eqref{eq:weights}. 
The following empirical simulated annealing algorithm outputs a feature assignment that yields very good normalized scores in a number of tested practical settings. 
\begin{table}[htb]
\centering
\begin{tabular}{l}
\hline
\textbf{Simulated Annealing Algorithm}\\
\textbf{Input:} Graph $G = (V,E)$, mixing parameter $\mu$, number of features $M$,\\
\qquad\quad number of rounds $N$;\\
\ 1: Let $A \equiv A_0 \colon V \to 2^{[M]}$ be an arbitrary feature assignment;\\
\ 2: \textbf{repeat}\\
\ 3: \quad Choose a vertex $u \in V$ uniformly at random;\\
\ 4: \quad Select $A'_u \subseteq \{1,2,\ldots,M\}$ uniformly at random;\\
\ 5: \quad Set $A'_v = A_v$ for all $v \neq u$;\\
\ 6: \quad Set $A = A'$ with probability $\min \{1,\exp\big(\mu(s(A')-s(A))\big)\}$;\\
\ 7: \textbf{until} the loop has run for $N$ rounds;\\
\textbf{Output:} The best observed assignment $A$, i.e., the one \\
which has the highest normalized score;\\
\hline
\end{tabular}
\end{table}

Extensive simulations with the above algorithm seem to suggest that 
setting $\mu = M$ offers best performance for a wide range of network topologies. 
The number of rounds $N$ that ensures quality results is $\mathcal{O}(n\log(n))$. 

Note that calculating $s(A)$ requires roughly $\O(n^3)$ operations. 
Therefore, one should compute $s(A)$ only once at the start of the algorithm. 
At every iteration when a candidate feature set $A'_u$ is generated, to compute $s(A')$, 
one should use the formula
\[
s(A') = s(A) + s_u(A') - s_u(A), 
\]
where $s_u(A)$ comprises the terms in \eqref{eq:SA} that involve $u$. 
There are $\binom{n-1}{2} + n - 1$ such terms. Therefore, in each iteration,
the computational complexity scales as $\O(n^2)$. Jointly with the preprocessing step, the annealing algorithm therefore has total 
time complexity $\O(n^3\log(n))$. 
\vspace{-0.05in}
\subsection{Upper Bounds on the Edge-Triangle Clique Cover}

An upper bound on the number of features, or equivalently, an upper bound on the edge-triangle clique cover number, may be used to guide the choice of the input parameter $M$ of the annealing algorithm (See the Simulation results section for a discussion of this issue). 
To determine a tight bound on the edge-triangle clique cover number, we recall a classical result from graph theory~\cite{ErdosGoodmanPosa1966}, which states that the edge clique cover number $\tetG$
satisfies the following inequality: 
\[
\teG \leq \left\lfloor \frac{n^2}{4} \right\rfloor,
\] 
for any graph $G$ on $n$ vertices. Equality is met when $G$ is the Tur\'{a}n graph $T(n,2)$~\cite{Turan1941}, a complete bipartite graph with one part consisting of $\lfloor n/2 \rfloor$ vertices and the other part consisting of $\lceil n/2 \rceil$ vertices. 
Next, we establish a nontrivial extensions of this result for $\tetG$.
\begin{theorem} 
\label{thm:upper_bound}
For any graph $G$ on $n \geq 7$ vertices, one has
\begin{equation}
\label{eq:upper_bound}
\tetG \leq 
\begin{cases}
\dfrac{n^3}{27}, &\text{ if } n \equiv 0 \pmod 3,
\vspace{0.1in} \\
\dfrac{(n-1)^3}{27} + \dfrac{(n-1)^2}{9}, &\text{ if } n \equiv 1 \pmod 3,
\vspace{0.1in} \\
\dfrac{(n+1)^3}{27} - \dfrac{(n+1)^2}{9}, &\text{ if } n \equiv 2 \pmod 3. 
\end{cases}
\end{equation}
\end{theorem} 
\begin{proof} The proof of Theorem~\ref{thm:upper_bound} is rather involved, and a sketch of the arguments is presented in the Appendix~\ref{app:erdos}. 
\end{proof}

\begin{theorem} 
\label{thm:Turan}
A graph of order $n$ has the edge-triangle clique
cover number $\tet$ attaining the upper bound given in Theorem~\ref{thm:upper_bound}
if and only if it is the Tur\'{a}n graph $T(n,3)$, a complete tripartite graph where the sizes of the parts differ from each other by at most one.
\end{theorem} 

This general purpose bound may be improved for a number of families of graphs, and in particular for complements of sparse graphs~\cite[Lemma~3.2]{Alon1986}, as stated in our next theorem. 

\begin{theorem}
\label{pr:Alon}
If $\deg(v) \geq n - d$ for every vertex $v$ of $G$, a graph of order $n$, where $d \geq 1$, then 
$\tetG \leq \lceil 3e^3(d+1)^3\log_e n \rceil$. 
\end{theorem}
\vspace{-0.1in}
\section{Simulation Results} \label{sec:simulations}

We tested both the MMCC algorithm with different choices of the motif weights as well as the simulated annealing approach with a number of clusters upper bounded according to Theorem~\ref{thm:upper_bound} on two small scale networks, in order to be able to discuss in detail various community structures that arise due to motifs (e.g., triangle). 

In the former case, we always set the similarity weight of edges and triangles to $1$,
and only tune the dissimilarity weight of nonedges and the relevance factor of triangles $\lambda_2=\lambda$. 
Different dissimilarity weights give different ``clustering resolutions'': Increasing the dissimilarity weight clearly leads to small clusters conglomerating into larger clusters. 

In the later case, the main challenge is to determine the correct choice for $M$, as it effectively represents the number of clusters. Most approaches rely on using a fraction of the edges for training and the remaining edges for actual community testing. We may also use an input parameter $M$ based on the theoretical upper bound of Theorem 2, scaled depending on the resolution of the communities we want.

The first network considered was described in~\cite{Palla2005}, comprising \emph{four overlapping social communities} that exhibit a number of triangle subgraphs. We first tested the MMCC method on this network, with edges and triangles treated as motifs, and we ran the approximation algorithm for two different choices of edge/nonedge and triangle/nontriangle weights. 
In the first test, we set the dissimilarity weights $1-w_e$ of the nonedges to lie in the interval $[1/2- 0.9 \times \epsilon, 1/2 - 0.5 \times \epsilon]$, where $\epsilon$ denotes the edge density of the network, defined as $\epsilon=|E|/\binom{n}{2}$. We kept the dissimilarity weight close to the value $1/2$ to account for the lack of influence of the nonedges on the community structures, but still strictly below $1/2$ in order to allow for more flexibility in the vertex placement procedure. The similarity weight of edges was set to $1$. Furthermore, we let the relevance factor of triangles, $\lambda$, range from $0$ to $50$, and set the similarity weight of triangles to $1$ and that of nontriangles to $1/2$. For all triangle relevance values $\lambda$ in the range $0-0.1$, which are very small, we recovered the original four communities of~\cite{Palla2005}, as triangles effectively played no role in the community structure. The results are depicted in Figure~\ref{Palla_Pan}. For all triangle relevance values in the range $0.2-9$ we obtained the same clustering result, comprising three communities, as depicted in Figure~\ref{Palla_Pan}. This clustering differs from the original structure outlined in~\cite{Palla2005} in so far that two clusters were joined into one (colored pink, involving vertices labeled starting with $7$). This is a consequence of the fact that a large number of triangles were crossing the two clusters, and with an increased relevance value of triangles, these motifs were grouped together. As expected, by making $\lambda$ very large - say, a value between $10$ and $50,$ we obtain one single cluster, as all triangles cluster together.
 
Applying the overlapping clustering method based on simulated annealing on the same network results in the same structure as reported in~\cite{Palla2005}, including four communities, except for one slight change: Node $17$ now belongs to two different communities instead of just one, as illustrated in Figure~\ref{Palla_Overlapping}. The explanation behind this result is that since node $17$ creates a triangle with both nodes $18$ and $19$, and the triangle motif encourages
these three nodes to lie within the same cluster, which also appears to more realistically explain the community structure. Note that in the simulations, we set $M=4$ to fairly compare our findings with those of~\cite{Palla2005}. The edge-triangle intersection number for the graph equals $16$\footnote{The bound of Theorem 2 equals $450$, which is roughly an order of magnitude larger.}, and using $M$ closer to this value would recover finer resolution community structures.
\begin{figure}[htb]
\centering
\vspace{-0.15in}
\includegraphics[scale=0.22]{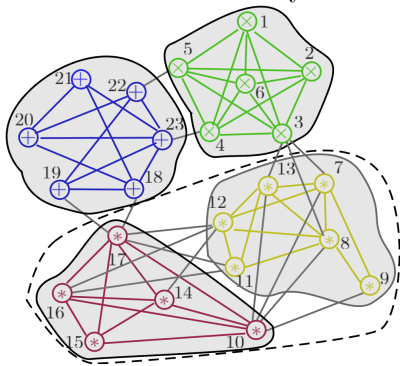}
\vspace{-0.06in}
\caption{Four and three nonoverlapping communities of the network~\cite{Palla2005}, obtained using the MMMC method.}
\label{Palla_Pan}
\vspace{-0.07in}
\end{figure}
The second example we present is the well studied Zachary Karate Club network~\cite{Zachary1977}. In the MMCC setting, we used the following parameter values: For the first set of tests, the dissimilarity weight of nonedges was set to $1/2 - 0.25 \times \epsilon$. The triangle similarity weight was set to $1$, and the relevance factor $\lambda$ kept in the range $1-4$. In this case, we found three, rather than the two original clusters, as node $10$ was placed in a cluster by itself (see Figure~\ref{KarateMMCC}). The reason behind this result is that the dissimilarity cost deviates significantly from the neutral value $1/2$ and there are a few connecting edges between $10$ and other nodes in the network. Node $10$ also does not close any triangles. For the second test, we set the dissimilarity weight of nonedges to be $1/2- 0.2 \times \epsilon$. In this case, we recovered the two ground truth clusters, with one mistake again relating to node $10$ which is now placed in a different cluster (see Figure~\ref{KarateMMCC}, where the node is marked by a dashed circle). The reason behind this classification is that the dissimilarity weight of nonedges is neutral, and that there are no triangle involving node $10$, so that $10$ is placed into the smaller of the two clusters.
\begin{figure}[htb]
\centering
\vspace{-0.1in}
\includegraphics[scale=0.25]{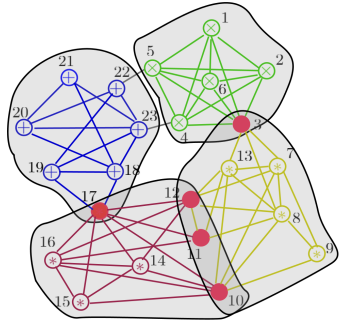}
\vspace{-0.52in}
\caption{An example of a network with four overlapping communities from~\cite{Palla2005}. Red nodes belong to overlapping communities.}
\label{Palla_Overlapping}
\vspace{-0.1in}
\end{figure}
The annealing algorithm with $M=2$ also recovers the two communities in the network (see Figure~\ref{Karateoverlapping})\footnote{The edge-clique number of the graph equals $34$, while Theorem 2 provides a rather loose upper bound of $1455$.}, except that now node $3$ belongs to both communities, as this node is not only well connected to both sides, but also closes a triangle
with both node $9$ and node $33$ in the left cluster. The edge-only version of the annealing algorithm~\cite{TsourakakisWWW2015} 
always misclassifies node $10$ by putting it into the right cluster, and it cannot find any
overlapping clusters. For a large range of values of the annealing parameter $M$, our method also puts node $3$ and node $34$ into two clusters simultaneously.

\vspace{-0.1in}
\section{Appendix} \label{sec:app}

\subsection{Proof of MMCC Hardness} \label{app:MMCCHard}

It is easy to see that the problem is in NP. To prove the claim, we focus our attention on the unweighted case $w_{K} \in\{0,1\}$ and use a reduction from the NP-complete Partition into Triangles problem.

Since $w_{K}\in\{0,1\}$, for simplicity of terminology we refer to a triplet $K$ with $w_{K}=1$ (respectively, $w_{K}=0$) as ``positive'' (respectively, ``negative''). We also use the term ``positive error'' to indicate that a positive triplet is placed across clusters and ``negative error'' to indicate that a negative triplet is placed within one cluster. Given a not necessarily complete graph $G=G(V,E)$, containing $n$ vertices where $n$ is a multiple of $3$, one problem of interest is to determine whether it can be partitioned into triangles. This problem, known as Partition into Triangles (PiT), has been proved to be NP-complete. To address the issue of MMCC hardness, we will exhibit a reduction of the PiT problem to the MMCC. As the first step in our proof, we construct a weighted graph $G^{w}$ that has the same vertex set as $G$. We set to $1$ the weights of triplets $G^{w}$ that correspond to triangles in $G$, and set the weights of all other triplets in $G^{w}$ to $0$. If there were a polynomial-time algorithm for the MMCC problem with the additionally imposed constraint that the size of each cluster is at most $3$, then we would be able to efficiently partition $G$ into triangles, a contradiction. As the MMCC algorithm does not necessarily generate clusters with bounded size, in what follows we describe how to construct another weighted graph, $H^{w},$ such that the MMCC algorithm applied on $H^{w}$ results in a bounded cluster-size run of MMCC on $G^{w}$.

The basic idea behind our approach is to impose the constraint on the size of clusters in $G^{w}$ by 
adding a large number of vertices into $H^{w}$ for each triplet in $G^{w}$, and then making the triplets 
inside the added vertices positive and other triplets negative. 
In this setting, a cluster in the new graph $H^{w}$ with more than $3$ vertices in $G^{w}$ causes too many negative errors and hence cannot be part of the optimal clustering. 

We now describe now how to construct a graph $H$ from $G$.
In addition to the vertices of $G$, for every triplet $\{j_1,j_2,j_3\}$ in $G$, $H$ contains additional $n^5$ vertices within a clique which we denote by $C_{j_1j_2j_3}$. Hence, $H$ contains $n+n^5{n\choose 3}$ vertices, and its edges include all edges inherited from $G$ along with the edges in the cliques and a set of edges fully connecting $\{j_1,j_2,j_3\}$ and $C_{j_1j_2j_3}$ (The vertices in the clique $C_{j_1j_2j_3}$ are not connected to any vertices inherited from 
$G$ other than $\{j_1,j_2,j_3\}$). It is also straightforward to show that $H$ has ${n\choose 3}\left[{n^5 \choose 2}+3n^5\right]+|E(G)|$ edges.
We use the term \emph{added sets} of a vertex $v \in H$ inherited from $G$ to refer to the vertices of the added cliques that contain $v$ as subscript; a similar terminology is used to refer to cliques containing pairs of vertices inherited from $G$. Clearly,
each vertex has ${n -1 \choose 2}$ corresponding added sets, while each pair of vertices has ${n-2 \choose 1}$ added sets. The weights of triplets of $H^{w}$ are determined as follows: Triplets comprising vertices from $G^{w}$ only have the same weights as those assigned in $G^{w}$; the weights of the remaining triplets, comprising vertices from $C_{j_1j_2j_3}\cup\{j_1,j_2,j_3\}$, have weight one, while all other triplets have weight zero.

Consider now a clustering $\mathcal{C}^*$ of $H^{w}$ of the following form:
\begin{enumerate}
\item There are ${n \choose 3}$ \emph{nonoverlapping} clusters.
\item Each cluster contains exactly one clique $C_{j_1j_2j_3}$ and potentially a subset of the corresponding three vertices $\{j_1,j_2,j_3\}$.
\item Each vertex in $V(G)$ lies in exactly one cluster that contains one of its corresponding added sets.
\end{enumerate}
In the above clustering, there are no errors arising due to triplets that lie across different added sets, since each cluster contains exactly 
one added set and the weights of triplets that lie across two clusters are equal to zero. The only errors arise from triplets with vertices contained in $V(G)$ or those involving both the vertices of $V(G)$ and the added sets. In the former case, the number of errors is at most ${n \choose 3}$. In the later case, each vertex in $V(G)$ is clustered together with just one added set and thus the number of positive errors induced by this vertex and its other corresponding added sets is exactly ${n^5 \choose 2}({n-1\choose 2}-1)$. For any pair of vertices in $V(G)$, the number of positive triplets that contain this pair and a vertex in the corresponding added sets of the pair is not larger than $n^5(n-2)$. Hence, the total number of errors for the described clusters is not larger than $n{n^5 \choose 2}({n-1\choose 2}-1)+{n \choose 2}n^5(n-2)+{n \choose 3}\sim\Omega(n^{13})$.

The clustering $C^*$ essentially partitions the vertices of $G$ into many small subsets, each of which containing at most three vertices. In our subsequent derivation, we show that the number of errors in a clustering that contains one cluster with at least four vertices from $V(G)$ must be larger than the number of errors induced by $C^*$.

First, observe that a clustering with fewer errors than $C^*$ has to have the size of each of its cluster lie in the interval $[n^5-n^4, n^5+n^4]$. 
Suppose that on the contrary there exists a cluster containing more that $n^5+n^4$ vertices. Then, there are at least ${n^5 \choose 2}n^4\sim\Omega(n^{14})$ errors caused by negative triangles across two different added sets within this cluster. Furthermore, each cluster must contain at least $n^5-n^4$ vertices of a clique, otherwise there are at least ${n^5 \choose 2}n^4\sim\Omega(n^{14})$ positive errors generated by ``splitting'' the corresponding added set. Since the size of each cluster is smaller than $n^5+n^4$, for each vertex in $V(G)$, the number of positive errors of the triplets formed by this vertex and two other vertices in the corresponding added sets of this vertex is lower bounded by ${n^5 \choose 2}{n-1\choose 2}-{n^5 \choose 2}-{n^4 \choose 2}$. 

Assume now that there exists a cluster that contains four vertices, say $\{j_1,j_2,j_3,j_4\}$, in $V(G)$. Then, there exists at least one vertex in $\{j_1,j_2,j_3,j_4\}$, say $j_1$, and at least $n^5-n^4$ other vertices that do not lie in a added set of $j_1$. Hence, the number of negative errors within this cluster is at least ${n^5-n^4 \choose 2}$. The total number of errors induced by such a clustering is therefore at least $n{n^5 \choose 2}({n-1\choose 2}-1)-n{n^4 \choose 2}+{n^5-n^4 \choose 2},$ 
which is larger than the number of errors in the clustering $\mathcal{C}^*$, for $n$ sufficiently large. Therefore, the optimal triangle-clustering has to be of the form of $\mathcal{C}^*$, imposing a constraint on the size of clusters in $G^{w}$.

\subsection{Proof of MMCC Approximation Guarantees} \label{app:approx}

Throughout the section, we use $\mathcal{V}$ to denote the set of unclustered vertices in one iteration. 

\begin{figure}[htb]
\centering
\vspace{-0.1in}
\includegraphics[scale=0.43]{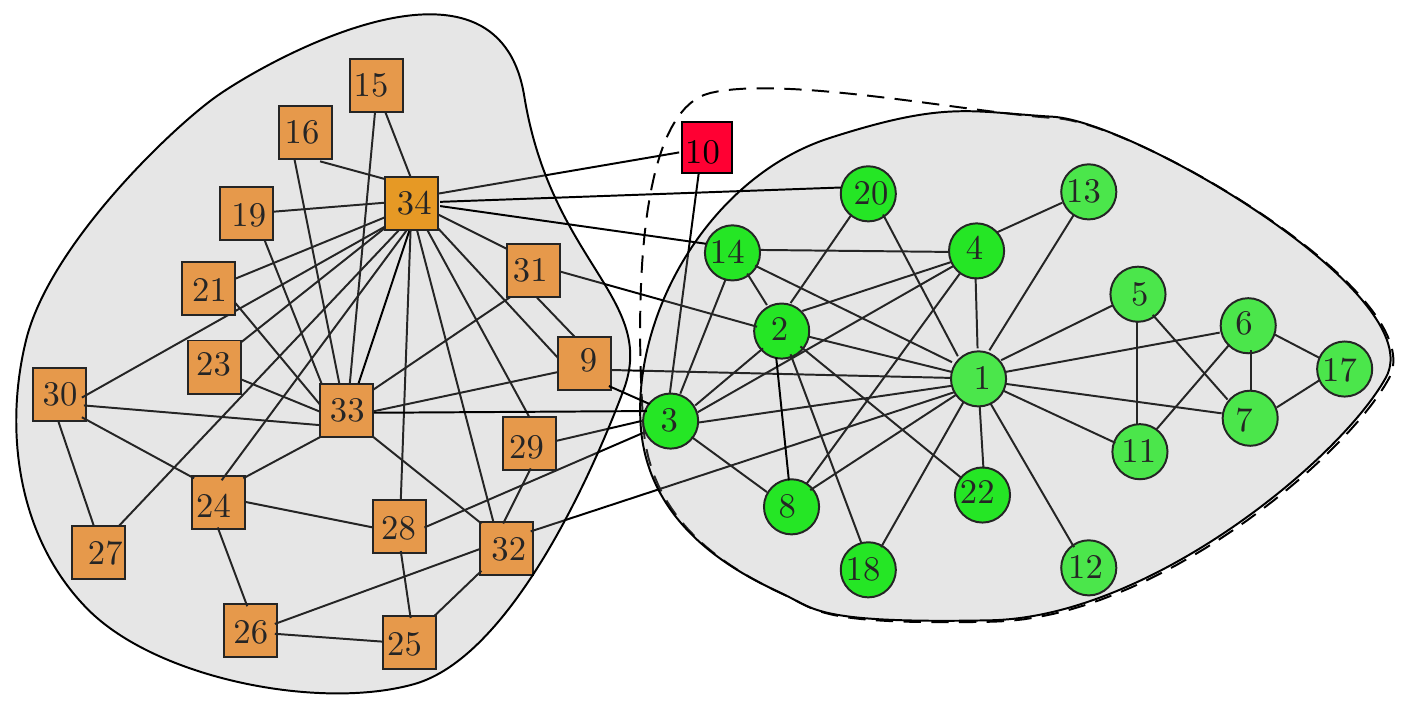}
\caption{Nonoverlapping communities for the Zachary Karate Club network~\cite{Zachary1977}, obtained using the MMCC method.} \label{KarateMMCC}
\vspace{-0.2in}
\end{figure}
\begin{figure}[htb]
\centering
\includegraphics[scale=0.48]{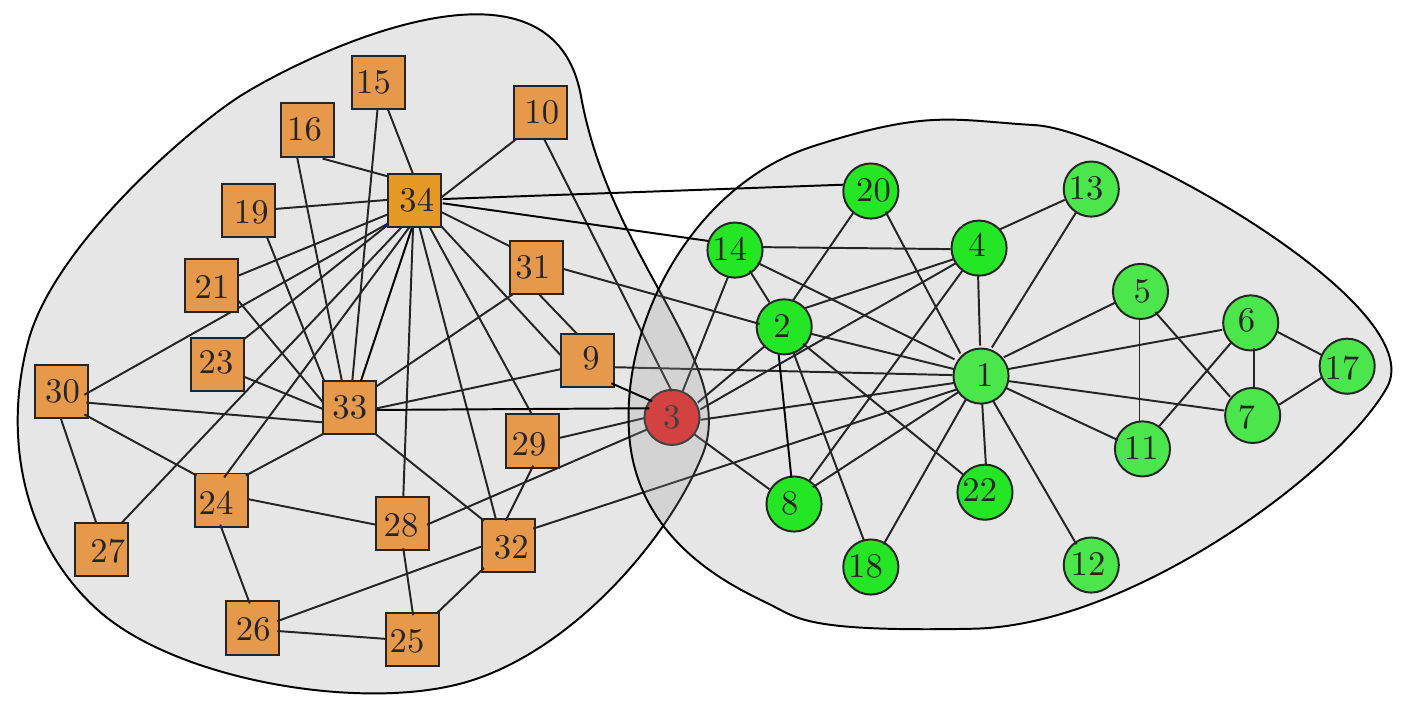}
\caption{Two overlapping communities detected in Zachary's Karate network~\cite{Zachary1977}.}  \label{Karateoverlapping}
\vspace{-0.11in}
\end{figure}
The proof to follows also often uses some immediate consequences of the LP constraints; it adopts the convention that $x_{jj} = 0$ for all $j \in V$:
\begin{enumerate}
\item $x_{jg}\geq x_{ij}-x_{ig}$ for any $i,j,g\in V$;
\item $x_{T}\geq \max_{j,g \in T}[x_{ij}-x_{ig}]$ for any $i\in V$;
\item $x_{T} \leq \frac{1}{2}\sum_{j,g\in T}x_{jg}\leq \frac{1}{2}\sum_{j,g\in T}(x_{ij}+x_{ig})\leq \sum_{j\in T}x_{ij}$ for any $i\in V$.
\end{enumerate}

When clustering splits (gathers) the endpoints of an edge or the vertices of a $3$-tuple into different clusters (in one single cluster), we call the result a \emph{break} (\emph{keep}). In each iteration of Algorithm 1, exactly one cluster will be output and thus break or keep the edges and $3$-tuples that have intersection with this cluster. To prove the rounding procedure can approximate the optimal solution within a constant factor $9$, it suffices to prove that for each iteration in Algorithm 1, the edges and $3$-tuples that are broken or kept will not increase the corresponding costs in the LP by more than $9$ times. We will do the analysis for the $3$-tuples only, since the analysis for edges is similar. 

Based on the output clusters being singletons or containing more vertices, we consider two different cases:

\caze{1} \emph{The output is a singleton cluster $\{i\}$}. 

 The clustering cost when outputting a singleton $\{i\}$ is $\sum_{T\subset \mathcal{T}(\mathcal{V}): i\in T}w_{T}$ while the LP cost is $\sum_{T\subset \mathcal{T}(\mathcal{V}): i\in T}(1-w_{T})(1-x_T)+w_{T}x_T$. 

 If $T\cap\{\mathcal{V}-(\{i\}\cup \Na(i))\}\neq \emptyset$, we have
 $x_{T}>\alpha$, so charging each such $3$-tuple $1/\alpha$
 times its LP-cost compensates for the cluster-cost. Therefore, it
 suffices to consider the $3$-tuples $T \in
 \mathcal{T}(\Na(i)\cup\{i\})$ with $i\in T$. Let $T=\{i,j,g\}$. Then, for any $j,g\in \Na(i)$ we have 
 $$\frac{1}{2}(x_{ij}+x_{ig})\leq x_{ijg}\leq x_{ij}+x_{ig}\leq 2\alpha,$$
 where the inequalities are based on the LP constraints. Hence, the LP cost of $T=\{i,j,g\}$ is bounded by
\begin{align*}         
&(1-w_{T})(1-x_{T})+w_{T}x_{T}\\
\geq&(1-w_{T})(1-x_{ij}-x_{ig})+\frac{1}{2}w_{T}(x_{ij}+x_{ig}) \\
 \geq &w_{T}\left[\frac{3}{2}(x_{ij}+x_{ig})-1\right]+(1-x_{ij}-x_{ig}).
\end{align*}
Since each $x_{ij}$ for $j \in T$ satisfies $x_{ij} \leq \alpha \leq 1/3$, the quantity
in square brackets is negative, so that $w_{ijg} \leq 1$ implies
\[
(1-w_{T})(1-x_{T})+w_{T}x_{T}\geq \frac{1}{2}(x_{ij}+x_{ig}).
\]
Summing over all $T=\{i,j,g\}$ such that $j,g \in \Na(i),\ j\neq g$, we see that
\begin{align*}
&\sum_{j,g \in \Na(i),j\neq g}[(1-w_{T})(1-x_{T})+w_{T}x_{T}] \\
\geq  &\sum_{j,g \in \Na(i),j\neq g}\frac{1}{2}(x_{ij}+x_{ig}) \geq \alpha\beta{\sizeof{\Na(i)} \choose 2}, 
\end{align*}
where the last inequality follows from the condition $\sum_{j \in \Na(i)}x_{ij} > \beta\alpha\sizeof{\Na(i)}$
that causes the algorithm to output $\{i\}$ as a singleton cluster.

Therefore, charging $1/(\alpha\beta)$ times the LP-cost to each
$3$-tuple that is kept or broken in Case 1 is enough to compensate for the total clustering cost
of these tuples.

\caze{2} \emph{The output is a cluster $\{i\}\cup \Na(i)$}. 

\emph{The cost of the $3$-tuples kept inside the cluster.}
The case $i\in T$ is the same as before: If $T=\{i,j,g\}$, then we have
$x_T\leq x_{ij}+x_{ig}\leq 2\alpha$, so charging
$1/(1-2\alpha)$ for this tuple is enough to compensate the
cluster-cost.

If $i\notin T$, order the vertices in $\Na(i)$ in
such a way that for any $j,g\in \Na(i)$, $j\prec g$
iff $x_{ij}< x_{ig}$ and assign an arbitrary order ($j\prec g$) when
the equality ($x_{ij}=x_{ig}$) holds.

For each vertex $g \in \Na(i)$, let $R_g=\{j \in
\Na(i) \st j\prec g\}$, and let $E_g$ be the set of
3-tuples $T$ such that $T\subset \Na(i)$ and $g$ is the largest vertex of $T$
according to $\prec$. (Thus, if $T \in E_l$, then $g \in E_g$ and $E_g \subset R_g \cup \{g\}$.)

Note that because of the order, we have $\sum_{j\in R_g} x_{ij}\leq
\alpha\beta|R_g|$.  Fix some $g \in \Na(i)$; we consider the
total cost of the $3$-tuples in $E_g$. The corresponding cluster-cost
is $\sum_{T \in E_g}1-w_{T}$ while the LP cost is $\sum_{T \in
  E_g}(1-x_{T})(1-w_{T})+x_{T}w_{T}$.

If $x_{ig} \leq \beta\alpha$, then for each
$T  \in E_g$, we have
\[ x_T \leq \sum_{j \in T}x_{ij} \leq 3x_{ig} \leq 3\beta\alpha, \] so
that charging $1/(1-3\beta\alpha)$ times the LP-cost to each
$3$-tuple in $E_g$ is enough to pay for the cluster cost of all such
tuples.

Now suppose that $x_{ig} > \beta\alpha$. In this case, for each $T  \in E_g$, we have $x_T \leq \sum_{j \in T}x_{ij}$,
hence $1 - x_T \geq 1 - \sum_{j \in T}x_{ij}$. Furthermore,
\[ x_T \geq \max_{j \in T-g}[x_{ig} - x_{ij}] \geq x_{ig} - \frac{1}{2}\sum_{j \in T-g}x_{ij}. \]
Letting $\sigma = \sum_{j \in T-g}x_{ij}$ so that $1-x_T \geq 1 - x_{ig} - \sigma$, we have
the following lower bound on the LP-cost of $T$:
\begin{align*}
  & (1-w_T)(1-x_T) + w_Tx_T \\
  \geq & (1-w_T)(1 - x_{ig} - \sigma) + w_T(x_{ig} - \frac{1}{2}\sigma) \\
  \geq &(1-w_T)(1 - 2x_{ig} - \frac{1}{2}\sigma) + x_{ig} - \frac{1}{2}\sigma. 
\end{align*}
Summing over all $T \in E_g$ and using the inequality $\sum_{T \in E_g}\frac{1}{2}\sum_{j \in T-g}x_{ij} \leq \sizeof{E_g}\beta\alpha$ yields the following lower bound on the total LP-cost of the
edges in $E_g$:
\begin{align*}
  &\sum_{T \in E_g}[(1-w_T)(1-x_T) + w_Tx_T] \\
  \geq &\sum_{T \in E_g}[(1-w_T)(1 - 2x_{ig} - \frac{1}{2}\sigma)+ x_{ig} - \beta\alpha] \\
  \geq &\sum_{T \in E_g}[(1-w_T)(1-x_{ig}-\frac{1}{2}\sigma - \beta\alpha)] \\
  \geq &\sum_{T \in E_g}\left[(1-w_T)(1-2\alpha - \beta\alpha)\right].
\end{align*}
Thus, charging each $3$-tuple in $E_g$ a factor of $1/(1-2\alpha-\beta\alpha)$ times its LP-cost
pays for the cluster-cost of all $3$-tuples in $E_g$.

\emph{The cost to break $3$-tuples across the cluster and the remaining part of $\mathcal{V}$.}  
As before, we call such tuples \emph{broken
  tuples}. Each broken tuple $T$ incurs a cluster-cost of $w_T$ and an
LP-cost of $x_Tw_T + (1-x_T)(1-w_T)$. First suppose that $T$ is a broken $3$-tuple with $i
\in T$, and $g \in T-\Na(i)\cup{i}$. Since $T$ is broken, we have $x_T \geq x_{ig}
> \alpha$, so charging $1/\alpha$ times the LP cost pays for
such $T$. We still must pay for the broken tuples $T$ with $i \notin T$.
For any set $T_1 \subset \mathcal{V}-\Na(i)$, $|T_1|<3$, let $E_{T_1}$ be the set of broken tuples $T$
such that $i \notin T$ and $T-\Na(i)= T_1$.  We show that the total
cluster-cost of the tuples in $E_{T_1}$ is at most a constant times their
total LP-cost.
First, suppose that there is some vertex $g \in T_1$ such that $x_{ig} \geq (1+\beta)\alpha$.
In this case, for every $T \in E_{T_1}$, we can take some arbitrary $j \in T \cap \Na(i)$ and obtain
\[ x_T \geq x_{ig} - x_{ij} \geq \beta\alpha, \]
since $j \in \Na(i)$ implies $x_{ij} \leq \alpha$. Thus, in this case,
charging $1/\alpha\beta$ times the LP-cost of each tuple in $E_{T_1}$ pays for the
cluster-cost of all tuples in $E_{T_1}$.

Next, suppose that $x_{ig} \leq (1+\beta)\alpha$ for all $g \in
T_1$. Consider any $T \in E_{T_1}$. Let $T_2 = T \cap \Na(i)$, let $\sigma_1 =
\sum_{j \in T_1}x_{ij}$ and let $\sigma_2 = \sum_{j \in
  T_2}x_{ij}$. We have the following bounds:
\begin{gather*}
 1 - x_{T} \geq 1 - \sum_{j \in T}x_{ij} = 1 - (\sigma_1 + \sigma_2), \\
 x_T  \geq \max_{g\in T_1,\ j \in T_2}[x_{ig} - x_{il}] \geq \frac{1}{\sizeof{T_1}}\sigma_1 - \frac{1}{\sizeof{T_2}}\sigma_2.
\end{gather*}
Combining these bounds yields the following lower bound on the LP-cost of $T$:
\begin{align}
  &(1-w_T)(1-x_T) + w_Tx_T \\
  \geq & (1-w_T)(1 - \sigma_1- \sigma_2) + w_T\left(\frac{\sigma_1}{\sizeof{T_1}} - \frac{\sigma_2}{\sizeof{T_2}}\right) \label{eq:ksquare-bound}\\
  = & w_T\left[\frac{\sizeof{T_1}+1}{\sizeof{T_1}}\sigma_1 + \frac{\sizeof{T_2}-1}{\sizeof{T_2}}\sigma_2 - 1\right] + 1 - \sigma_1 - \sigma_2\nonumber.
\end{align}
Using the bijection
between $E_{T_1}$ and ${\Na(i) \choose \sizeof{T_2}}$ and $\sizeof{T_2}=3-\sizeof{T_1}$, we see that $\sum_{T \in E_{T_1}}\sigma_2 \leq (3-\sizeof{T_1})\beta\alpha{\sizeof{\Na(i)} \choose (3-\sizeof{T_1})}$. Furthermore, since $\alpha, \beta \leq 1/3$, we have
\[ 1 - \sigma - \sizeof{T_1}\beta\alpha \geq 1 - \sizeof{T}(1+\beta)\alpha - \sizeof{T_1}\beta\alpha \]
\[ \;\;\;\;\;\;\;\;\;\;\;\;\;\;\;\;\;\;\;\;\;\;\;  \geq 1 - 2(1+\beta)\alpha - \beta\alpha \geq 0. \]
Therefore, summing the above inequality over all $T \in E_{T_1}$ gives the following
lower bound on the total LP-cost of all tuples in $E_{T_1}$:
\begin{align*}
  &\sum_{T \in E_{T_1}}[(1-w_T)(1-x_T) + w_Tx_T] \\
   \geq & \sum_{T \in E_{T_1}} w_T\left[\frac{\sizeof{T_1}+1}{\sizeof{T_1}}\sigma_1 + \frac{\sizeof{T_2}-1}{\sizeof{T_2}}\sigma_2 - 1\right] \\ 
           + &\sum_{T \in E_{T_1}} (1 - \sigma_1- \beta\alpha\sizeof{T_2}) \\
   \geq &\sum_{T \in E_{T_1}}w_T\left[\frac{1}{\sizeof{T_1}}\sigma_1 + \frac{\sizeof{T_2}-1}{\sizeof{T_2}}\sigma_2 - \beta\alpha\sizeof{T_2}\right]  \\
\geq & \sum_{T \in E_{T_1}}w_T(\alpha - 2\alpha\beta).
\end{align*}
As a result, charging a factor of $1/[\alpha(1-2\beta)]$ times the LP-cost
of each tuple in $E_T$ pays for the cluster-cost of all tuples in $E_T$.
 
In summary, if $\alpha, \beta \leq 1/3$, then charging each tuple a factor of $c$ times its LP cost,
where
\begin{align*}
c=\max\{\frac{1}{\beta\alpha}, \frac{1}{1-2\alpha}, \frac{1}{1-2\alpha-\beta\alpha},\frac{1}{\alpha(1-2\beta)}\} = \frac{1}{\beta\alpha},
\end{align*}
is enough to compensate the cluster-cost of all tuples. By setting $\alpha=\beta=1/3$, which minimizes $c$, we obtain $9$ as the approximation factor. 

\subsection{Proof of Hardness for Finding the Edge-Triangle Cover Number} \label{app:ETCHard}

It is obvious that ETCC is in NP. 
We prove the NP-completeness of this problem by establishing a reduction from the ECC problem,
which is know to be NP-complete~\cite{Orlin1977, KouStockmeyerWong78}.
Let $G$ be an arbitrary graph of order $n$ and let $M \geq 1$. Let $G'$ be the graph obtained from $G$ by introducing 
\begin{itemize}
	\item $s = 1 + \tetG$ new vertices $\{u_1,\ldots, u_s\}$, and
	\item $sn$ new edges that connect the new vertices to all existing vertices of $G$.
\end{itemize}
By Theorem~\ref{thm:upper_bound}, the graph $G'$ has order and size polynomial in $n$. 
Let $M' = sM + \tetG$. We demonstrate that $\te(G) \leq M$ if and only if $\tet(G') \leq M'$. 

Indeed, suppose that $\te(G) \leq M$, i.e. there is a set $\A$ of at most $M$ cliques in $G$
that collectively cover all edges in $G$. Then we can cover all edges and triangles in $G'$ by a set of cliques $\B$
obtained from $\A$ by adding each vertex in $\{u_1,u_2,\ldots,u_s\}$ to each clique in $\A$, together with a
minimum set of cliques $\C$ of $G$ that can cover all edges and triangles in $G$, 
which has size $\tetG$. In total, this cover has at most 
\[
s|\A| + |\C| \leq sM + \tetG = M'  
\]
cliques. Thus, if $\te(G) \leq M$ then $\tet(G') \leq M'$. Conversely, suppose that we 
have an edge-triangle clique cover $\D$ of $G'$ of size at most $M'$. Let $\D_i$ be the subset of cliques in $\D$ that contain the vertex $u_i$, for $1 \leq i \leq s$. As $u_i$ and $u_j$ are not adjacent, for $i \neq j$, $\D_i$ and $\D_j$ do not
have any common cliques. Hence,  
\[
\sum_{i = 1}^s |\D_i| \leq |\D| \leq M'.
\]
Therefore, if $i_{\min}$ is an index such that $|\D_{i_{\min}}| = \min_{1 \leq i \leq s} |\D_i|$, then 
\[
|\D_{i_{\min}}| \leq \left\lfloor \dfrac{\sum_{i = 1}^s |\D_i|}{s} \right\rfloor \leq \left\lfloor \dfrac{M'}{s}
\right\rfloor
=\left\lfloor \dfrac{sM + \tetG}{s} \right\rfloor = M,
\]
where the last equality holds because $s = 1 + \tetG$. 
Then, by removing $u_{i_{\min}}$ from all cliques in $\D_{i_{\min}}$, we obtain an edge clique
cover of $G$ of size at most $M$. The proof follows. 

\subsection{Proof of Upper Bound on The Edge-Triangle Cover Number} \label{app:erdos}

\begin{lemma} 
\label{lem:0}
Let $G$ be a graph on $n \geq 3$ vertices, and $(u,v,w)$ be a triangle in $G$. Then 
\begin{equation} 
\label{eq:tuvw}
\tuvw(G) \leq \lfloor\frac{n^2}{3}\rfloor-n+1,
\end{equation}
where $\tuvw(G)$ denotes the minimum number of cliques of $G$ that can cover all 
edges and triangles that contain at least one vertex among $u$, $v$, and $w$.  
\end{lemma} 
\begin{proof}[Sketch] 
We prove this lemma by induction on $n$. The inequality~\eqref{eq:tuvw} obviously holds for $n\leq 5$. We now assume that
$n \geq 6$ and that \eqref{eq:tuvw} holds for all graphs of order $n - 3$. 
We need to prove that this inequality also holds for a graph $G$ of order $n$. 
Let $G_{n-3}$ be the subgraph of $G = (V,E)$ induced by the set of vertices $V \setminus \{u,v,w\}$. We consider the following two cases.\\ 

\nin\textbf{Case 1.} 
The graph $G_{n-3}$ has no triangles. 
In order to bound $\tuvw(G)$, we analyze the edges and triangles of $G$ that contain at least
one vertex from $\{u,v,w\}$. There are three types of such edges and triangles.	
Type~1 consists of the edges and triangles that only involve $u$, $v$, and $w$. Obviously, we can
cover all of these edges and triangles by using just one clique $(u,v,w)$. 
Type~2 consists of the edges and triangles that contain precisely one vertex in $G_{n-3}$. 
For each vertex $x$ of $G_{n-3}$, since $u$, $v$, and $w$ form a triangle, we can use at most one clique to cover all 
edges and triangles of Type~2 associated to $x$. Therefore, we can cover all edges and triangles of Type~2 by 
		at most $n-3$ cliques. Type~3 consists of the triangles that contain precisely one vertex 
		from $\{u,v,w\}$, and two vertices in $G_{n-3}$. Similarly, we can use at most one clique to cover all 
	triangles of Type~3 containing each edge $(x,y)$ of $G_{n-3}$. Since $G_{n-3}$ is triangle-free
	by assumption, according to Tur\'{a}n's theorem~\cite{Turan1941}, $G_{n-3}$ has
	at most $(n-3)^2/4$ edges. Therefore, we can cover all triangles of Type~3 by at most $(n-3)^2/4$ cliques.

By summing up the number of cliques to cover the edges and triangles of all three types we obtain
\begin{equation}
\label{eq:case1}
1 + (n-3) + \dfrac{(n-3)^2}{4} \leq \dfrac{n^2}{3}-n+1 \Longleftrightarrow 0 \leq (n-3)^2,
\end{equation}
which is always true. Thus, in this case, \eqref{eq:tuvw} holds. 
\begin{figure}[htb]
\centering
\includegraphics[scale=0.55]{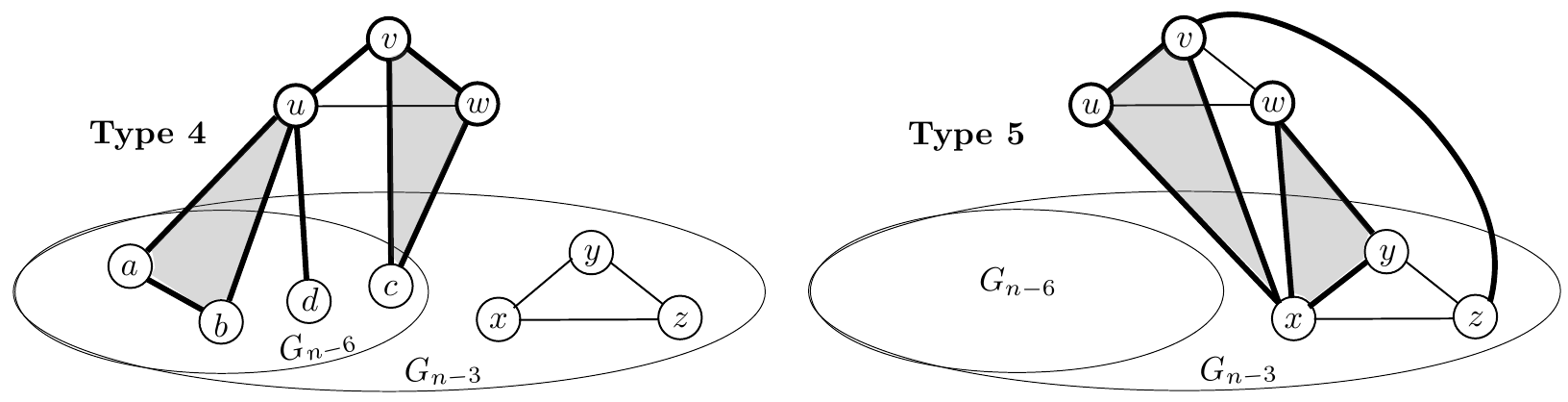}
\caption{(Case 2) Edges and triangles of Type~4, each of which contains at least one vertex from $\{u,v,w\}$, but no vertex from $\{x,y,z\}$, and of Type~5, each of which contains at least one vertex  from $\{u,v,w\}$, at least one vertex from $\{x,y,z\}$, and no vertex from $G_{n-6}$ (we ignore the triangle $(v,x,z)$ for clarity).}
\label{fig:types45}
\end{figure}

\nin\textbf{Case 2.} 
Let $(x,y,z)$ be a triangle of $G_{n-3}$ that forms the largest number of edges with $(u,v,w)$. In other words, we choose the triangle $(x,y,z)$ so that the number of edges $(p,q)$, where $p \in \{u,v,w\}$ and $q \in \{x,y,z\}$, is maximized among all triangles of $G_{n-3}$.  
Let $G_{n-6}$ be the subgraph of $G$ induced by the vertices $V \setminus \{u,v,w,x,y,z\}$. In order to bound $\tuvw(G)$, we analyze the edges and triangles of $G$ that contain at least one vertex from $\{u,v,w\}$. There are three types of such edges and triangles
(see Fig.~\ref{fig:types45} and Fig.~\ref{fig:type6}). 
Type~4 consists of the edges and triangles that contain at least one vertex from $\{u,v,w\}$, but no vertex from $\{x,y,z\}$. By the inductive hypothesis, we can cover all edges and triangles of Type~4 by at most 
	\begin{equation} 
	\label{eq:1}
	\dfrac{(n-3)^2}{3} - (n-3) + 1
	\end{equation} 
cliques. Type~5 consists of the edges and triangles that contain at least one vertex from $\{u,v,w\}$, at least one vertex from $\{x,y,z\}$, and no vertex from $G_{n-6}$. We can show that the edges and 
	triangles of Type~5 can be covered by at most \emph{six} cliques, without much difficulty.
\begin{figure}[htb]
\centering
\includegraphics[scale=0.55]{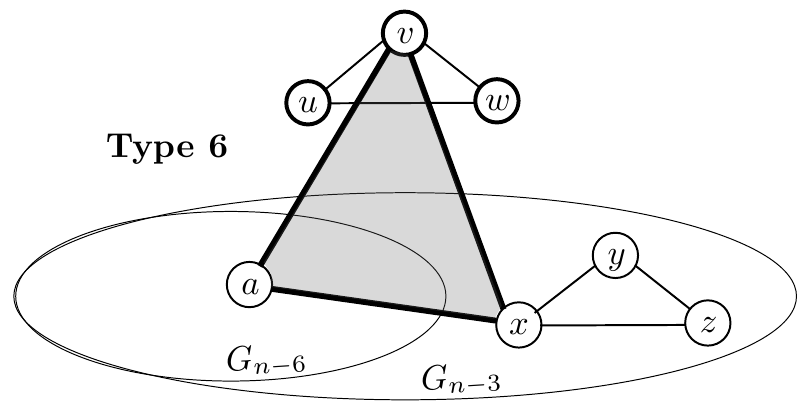}
\caption{(Case 2) Triangles of Type~6, each of which contains one vertex from $\{u,v,w\}$, one from $\{x,y,z\}$, and one from $G_{n-6}$.}
\label{fig:type6}
\end{figure}
Type~6 consists of the triangles that contain one vertex from $\{u,v,w\}$, one from $\{x,y,z\}$, and one from $G_{n-6}$. It can be shown that one can cover all triangles of Type~6 with at most $2(n-6)$ cliques. The key idea is to prove that we can use at most \emph{two} cliques to cover all triangles of this type that contain each fixed vertex $a$ of $G_{n-6}$. We omit the details.
Finally, the numbers of cliques used to cover all edges and triangles of Type~4, Type~5, and Type~6 sum up precisely to $n^2/3-n+1$.   
\end{proof} 

A proof similar to the one presented for Lemma~\ref{lem:0} may be used to prove Lemma~\ref{lem:12}.  

\begin{lemma} 
\label{lem:12}
Let $G$ be a graph on $n$ vertices, where $n \not\equiv 0 \pmod 3$, and $(u,v,w)$ is a triangle in $G$. Then 
\begin{equation} 
\label{eq:tuvw12}
\tuvw(G) \leq \frac{n^2+2}{3} - n,
\end{equation}
where $\tuvw(G)$ denotes the minimum number of cliques of $G$ that can cover all edges and triangles 
that contain at least one vertex among $u$, $v$, and $w$.  
\end{lemma} 

\begin{lemma}
\label{lem:base_case}
The inequality \eqref{eq:upper_bound} stated in Theorem~\ref{thm:upper_bound} holds for $n = 7, 8, 9$. 
\end{lemma}
\begin{proof}[Sketch]
For $n = 7,8,9$, we can apply the strategy used in the proof of Lemma~\ref{lem:0}, by taking out a triangle
$(u,v,w)$ of $G$, if any, and then considering two cases, depending on whether $G_{n-3} = G \setminus \{u,v,w\}$
contains a triangle $(x,y,z)$ or not. In both cases, we can show that \eqref{eq:upper_bound} holds. 
Note that if $G$ does not contain any triangles, then by \turan's theorem~\cite{Turan1941}, all edges of $G$ can
be covered by at most $\lfloor n^2/4\rfloor \leq \lfloor n/3 \rfloor \lfloor (n+1)/3 \rfloor \lfloor (n+2)/3 \rfloor$ cliques,
which are the edges themselves, for $n \geq 7$. We omit the remaining details.   
\end{proof}

\begin{proof}[Proof of Theorem~\ref{thm:upper_bound}]
We also prove this theorem by induction on $n$. 
The base case follows from Lemma~\ref{lem:base_case}. 

\nin\textbf{Induction step.}
 
Suppose that $n \geq 10$ and that the statement \eqref{eq:upper_bound} of the theorem holds for all graphs on $n-3$ vertices. 
We aim to prove that \eqref{eq:upper_bound} also holds for any graph $G = (V,E)$ on $n$ vertices.
If $G$ has no triangles then by \turan's theorem~\cite{Turan1941}, all edges of $G$ can be covered by at most 
$\lfloor n^2/4\rfloor \leq \lfloor n/3 \rfloor \lfloor (n+1)/3 \rfloor \lfloor (n+2)/3 \rfloor$ cliques (edges), for $n \geq 10$, 
and hence Theorem~\ref{thm:upper_bound} holds trivially. We now assume that there exists some triangle 
$(u,v,w)$ in $G$. Let $G_{n-3}$ be the subgraph of $G$ induced by the vertex set $V \setminus \{u,v,w\}$.
If $n \equiv 0 \pmod 3$, then by our inductive hypothesis, all edges and triangles in $G_{n-3}$ can be covered
by using at most $\frac{(n-3)^3}{27}$ cliques. Moreover, by Lemma~\ref{lem:0}, all edges and triangles in $G$ 
that contain at least one vertex from $\{u,v,w\}$ can be covered by at most $\frac{n^2}{3}-n+1$ cliques. 
	Thus, all edges and triangles in $G$ can be covered by using at most
	\[
	\tet(G_{n-3}) + \tuvw(G) \leq \dfrac{(n-3)^3}{27} + \Big(\dfrac{n^2}{3}-n+1 \Big) 
	= \dfrac{n^3}{27}
	\]
	cliques. Hence, Equation~\eqref{eq:upper_bound} holds for $G$ as well. The cases $n \equiv 1,2 \pmod 3$ can be handled similarly.
\end{proof}

\subsection{Proof of the Upper Bound on the Edge-Triangle Clique Cover Number for Complements of Sparse Graphs}
\label{app:alon}

Let $M = \lceil 3e^3(d+1)^3\log_e n \rceil$. Each set $C'_k$, $k = 1,2,\ldots, M$, is created independently
by including each vertex $v$ with a probability of $1/(d+1)$. Then for each $m$, let $C_m$ be obtained from
$C'_m$ by removing those vertices that have some non-neighbors in $C'_m$. Obviously $C_m$ is a clique of $G$.
We aim to show that the expected number of edges and triangles that are not contained in any clique $C_m$, $k = 1,2,\ldots,M$, 
is smaller than one, which implies that there exists an ETCC of size $M$. 

For each $m$, each triangle $(u,v,w)$ of $G$ is covered by $C_m$ if all three vertices are included in $C'_m$ and 
none of their non-neighbors are chosen. Therefore, the probability that $(u,v,w)$ is covered by $C_m$
is at least 
\[
\frac{1}{(d+1)^3}\Big(1-\frac{1}{d+1}\Big)^{3d} \geq \frac{1}{e^3(d+1)^3},
\]
where the inequality follows from the inequality $(1-\frac{1}{x})^{x-1} \geq \frac{1}{e}$, 
where $x = d + 1 \geq 2$. 
Therefore, the probability that $(u,v,w)$ is not covered in any $C_m$'s is at most
\[
\Big( 1 - \dfrac{1}{e^3(d+1)^3} \Big)^M \leq \exp(-\dfrac{M}{e^3(d+1)^3}) 
\]
\[
\;\;\;\;\;\;\;\;\;\;\;\;\;\;\;\;\;\;\;\;\;\;\;\;\;\;\;\;\;\;\;\;\;\;\;< \exp(-3\log_e n) = \dfrac{1}{n^3}, 
\]
where the first inequality is from the inequality $(1-\frac{1}{x})^y \leq \exp(-\frac{y}{x})$, for all $x > 1$ and $y > 0$, 
and the second one is because $M \geq 3e^3(d+1)^3\log_e n$. Hence, the expected number of triangles that are not covered by
any of the $C_m$'s is at most
\begin{equation}
\label{eq:Alon1}
\binom{n}{3}\dfrac{1}{n^3} < \dfrac{1}{6}.
\end{equation}
Since $M > 2e^2(d+1)^2$, the same computation shows that the expected number of edges that are not covered by
any of the $C_m$'s is at most
\begin{equation}
\label{eq:Alon2}
\binom{n}{2}\dfrac{1}{n^2} < \dfrac{1}{2}.
\end{equation}
From \eqref{eq:Alon1} and \eqref{eq:Alon2}, by the additivity of expectation, we deduce that the expected number of edges and triangles that are not cover by the cliques $C_m$'s, $k = 1,2,\ldots,M$, is smaller than one. 

\bibliographystyle{IEEETran}
\bibliography{TriangleCliqueCover,cluster,cluster1}
\end{document}